\documentclass[letterpaper, 10 pt, conference]{ieeetran}  % Comment this line out if you need a4paper

\IEEEoverridecommandlockouts                              % This command is only needed if 
\newcommand{\norm}[2]{\big\Vert #2 \big\Vert_{#1}}

\newcommand{\ldlt}{\text{$\mathrm{LDL^\top}$}}

\newtheorem{theorem}{Theorem}

\newtheorem{corollary}[theorem]{Corollary}
\newtheorem{definition}{Definition}

\usepackage{subcaption}
\usepackage{floatrow}
\usepackage{graphicx} % for pdf, bitmapped graphics files
\usepackage{comment}
\usepackage{cite}
\usepackage{algorithm}
\usepackage{algorithmic}
\usepackage{booktabs}

\usepackage{amsmath} % assumes amsmath package installed
\usepackage{amssymb}
\usepackage{bm}
\usepackage{balance}
\usepackage{hyperref}
\hypersetup{
    colorlinks=true,
    linkcolor=black,
    filecolor=blue,      
    urlcolor=black,
    pdftitle={Paper},
    pdfpagemode=FullScreen,
    }

\usepackage{xcolor}
\floatstyle{plaintop}
\restylefloat{table}
\usepackage[tableposition=top]{caption}

\title{\Large \bf
PDO-Split: A Real-Time Solver for Constrained Potential-Dominant Games
}

\author{Zhiyuan Zhang \qquad Yue Guan \qquad Panagiotis Tsiotras$^{1}$% <-this % stops a space
\thanks{This work is sponsored by ONR awards N00014-23-2308 and N00014-23-2353}
\thanks{$^{1}$ School of Aerospace Engineering, Institute for Robotics and Intelligent Machines, Georgia Institute of Technology, Atlanta, GA, USA. Emails:
        {\tt\small \{zzhang615, yguan44, tsiotras\}@gatech.edu}}
}
\renewcommand{\baselinestretch}{0.98}

\begin{document}
\maketitle
\thispagestyle{empty}
\pagestyle{empty}

\begin{abstract}
    We present Potential-Dominant Operator Splitting (PDO-Split), a real-time generalized Nash equilibrium solver for constrained dynamic multi-agent problems with potential-dominant structure.
    PDO-Split targets the computational bottleneck in game-theoretic model predictive control, where repeatedly solving the equilibrium KKT system becomes costly as the number of agents increases.
    We introduce a class of \emph{$\gamma$-potential-dominant games}, in which the potential (or cooperative) structure dominates the competitive interactions between the agents. 
    We show that this structure allows PDO-Split to use the potential component as a preconditioner while incorporating the competitive component through iterative refinement.
    We establish convergence guarantees and exploit the symmetry of the potential component through a reusable \ldlt{} factorization to accelerate the linear solvers.
    We evaluate PDO-Split across several practical multi-agent planning problems, including eight-car racing and six-vehicle ramp merging, achieving substantially improved runtime and convergence rate over existing DDP-based and Newton-based methods, especially with a larger number of agents. 
    The algorithm is also experimentally validated on miniature autonomous race car platforms, demonstrating the approach's real-time capability. 
\end{abstract}

%%%%%%%%%%%%%%%%%%%%%%%%%%%%%%%%%%%%%%%%%%%%%%%%%
\section{Introduction}
%%%%%%%%%%%%%%%%%%%%%%%%%%%%%%%%%%%%%%%%%%%%%%%%%

\IEEEPARstart{M}{odern} robotic applications, such as warehouse automation~\cite{stern2019multi, li2021lifelong} and autonomous driving~\cite{algames, schwarting2019social}, require many agents to reason and react in highly coupled multi-agent environments.
As each agent’s decisions affect the performance and the constraints of the other agents, these environments are highly dynamic and require decision-making at sufficiently fast timescales to ensure adaptability, safety and performance.
Real-time multi-robot planning therefore demands algorithms that are capable of handling complex strategic interactions while being computationally efficient.

A common solution concept for constrained multi-agent dynamic games is the Generalized Nash Equilibrium (GNE).
A GNE is a set of control strategies for all agents such that no agent can unilaterally improve its own cost without violating the shared constraints.
Real-time numerical GNE solvers are fundamental for enabling MPC-style receding horizon control in interactive multi-agent environments.
However, existing methods~\cite{algames, ilqgame, rd3g} often remain too slow for real-time implementation and scale poorly with the number of agents.
A fundamental source of this difficulty is the agents' competing objectives, which can induce cycling behavior in the best-response dynamics~\cite{balduzzi2019open}, causing standard optimization methods to converge slowly, stagnate, or fail altogether.
In contrast to competitive games, potential games~\cite{monderer1996potential} represent agents’ incentives by a common potential function.
A simple example of a potential game is a fully cooperative game.
This structure reduces equilibrium computation to an optimization problem closely resembling single-agent control and enables the use of efficient numerical optimization methods.

Most real-world interactions, however, are neither purely competitive nor exactly potential.
Instead, they are often a mixture of cooperative interactions with weaker competitive incentives, where the cooperative/potential component dominates.
For example, in highway merging, all vehicles share the primary objective of avoiding collisions, maintaining speed, and minimizing control effort, while competing only secondarily over relative position or travel time.
Such a potential-dominant structure creates an opportunity to develop fast game-theoretic planning algorithms that exploit the alignment among agents while treating strategic conflicts as disturbances. 

In this work, we introduce a class of \emph{$\gamma$-potential-dominant games} and characterize it through the spectral properties of the gradient dynamics.
We solve for the GNE as a residual minimization problem and decompose the gradient operator into a symmetric potential component and a skew-symmetric competitive component, and use the former as a preconditioner.
An \ldlt{} factorization of the potential component is implemented using CasADi~\cite{casadi} to further accelerate the implementation by reusing the factorization and the differentiation of the computation graph. 
Physical experiments and numerical simulations demonstrate substantial improvements in computational speed across a range of multi-agent scenarios.

%%%%%%%%%%%%%%%%%%%%%%%%%%%%%%%%%%%%%%%%%%%%%%%%%
\section{Related Work}
%%%%%%%%%%%%%%%%%%%%%%%%%%%%%%%%%%%%%%%%%%%%%%%%%

We focus on constrained dynamic games and solvers for Generalized Nash Equilibria (GNE). 
A GNE is a set of control strategies such that no agent can improve its performance by unilaterally deviating without violating constraints. 

\subsection{Generic Game Solvers}
Existing GNE solvers for constrained dynamic games can be broadly categorized into DDP-based and Newton-based methods.

\paragraph{DDP-based methods}

These methods extend Differential Dynamic Programming (DDP) to multi-agent settings.
Minimax DDP~\cite{minimax_ddp} formulates each stage as a minimax game, whereas OPTGame3~\cite{optgame} and iLQGames~\cite{ilqgame} extend this approach to general-sum games.
Starting from a nominal trajectory, these methods construct a local linear-quadratic (LQ) approximation of the dynamic game, solve the resulting LQ games through a backward optimality pass, and update the nominal trajectory through a forward rollout using the resulting control policies.
The algorithm iterates between the backward and forward passes until convergence.
While the standard LQ-game formulation does not directly accommodate constraints, they can be incorporated through barrier functions~\cite{ilqgame_barrier1,ilqgame_barrier2} or augmented-Lagrangian formulations~\cite{ilqgame_lagrange}.
These methods have been applied to real-time robotic problems, including drone racing and traffic interactions~\cite{ilqgame_racing2}.
However, their repeated forward and backward passes can become computationally expensive and scale poorly as the number of agents and the strength of their interactions increase~\cite{rd3g}.

\paragraph{Newton-based methods}

Newton-based methods directly solve the coupled first-order optimality (KKT) conditions of the game.
Both single-shooting~\cite{dldu_1,dldu_2} and multiple-shooting formulations~\cite{algames,rd3g} have been considered.
Single-shooting considers only control inputs as optimization variables and eliminates state variables via forward simulations.
Multiple shooting, on the other hand, retains both states and controls as optimization variables and enforces the dynamics as equality constraints.
Although multiple shooting introduces more decision variables, it preserves temporal sparsity and, generally, offers greater numerical robustness and faster convergence~\cite{algames}.

A prominent multiple-shooting solver is ALGAMES~\cite{algames}, which combines a Newton-type method with an augmented-Lagrangian treatment of the constraints.
ALGAMES demonstrated shorter solution times than the DDP-based iLQGames solver in benchmark problems involving intersection traversal and highway merging~\cite{rios2016automated}.
More recently, RD3G~\cite{rd3g} further improved computational performance by handling active-set constraints and demonstrated real-time receding-horizon control on physical autonomous racing platforms.

Despite these advances, Newton-based methods treat the game as a generic coupled nonlinear system and therefore require repeated solutions involving the full game KKT matrix.
The cost of these operations grows rapidly with the number of agents and the planning horizon.
Moreover, competitive interactions introduce asymmetric and rotational components into the gradient dynamics, which can lead to poor conditioning, oscillatory updates, and slow convergence.
These limitations motivate the exploitation of additional game structure to achieve faster convergence.

\subsection{Near-Potential Game Solvers}

A structure commonly exploited to simplify equilibrium computation is the potential-game framework~\cite{monderer1996potential}.
In a Markov potential game, the change in each agent's long-term utility resulting from a unilateral policy deviation is captured by a common potential function~\cite{macua2018learning}.
Equilibrium computation can therefore be reduced to optimizing this potential function, making the problem resemble a single-agent optimization problem rather than a coupled equilibrium problem~\cite{maheshwari2025independent, maheshwari2024convergence}.

Recent work has extended this framework to near-potential games~\cite{guo2025markov, candogan2011flows}, in which changes in the agents' utilities are represented by a potential function up to an approximation error.
This near-potential structure permits the use of single-agent optimization methods, such as gradient descent, while providing bounds on the resulting Nash gap.

Although optimization can be efficient once a potential function is available, identifying such a function may itself require additional effort.
Learning-based approaches~\cite{kalaria2024alpha} provide a promising way to construct approximate potential functions, but their performance may depend on the amount of available data, training time, and problem-specific design choices.
Since the learned potential only approximates the original game, the resulting solution generally corresponds to an approximate Nash equilibrium, with accuracy determined jointly by both the potential-approximation error and the optimization error.

In this work, we introduce a class of $\gamma$-potential-dominant games, defined through the spectral properties of the local gradient dynamics.
Rather than constructing or learning a global potential function, our approach uses local gradient information to quantify the dominance of the potential component over the competitive component.
We show that this structure can be exploited to develop faster algorithms that solve the original game directly, without requiring explicit knowledge of a global potential function.

%%%%%%%%%%%%%%%%%%%%%%%%%%%%%%%%%%%%%%%%%%%%%%%%%
\section{Problem Formulation}
%%%%%%%%%%%%%%%%%%%%%%%%%%%%%%%%%%%%%%%%%%%%%%%%%

We consider a discrete-time $N$-agent game. 
For each agent $i\!\in\!\{1, \ldots,N\}$, the vectors $x^i_t \!\in\! \mathbb{R}^n$ and $u^i_t\!\in\!\mathbb{R}^m$ represent the state and the control input of agent $i$ at stage~$t$, respectively. 
For notational brevity, we adopt the following convention:
$x_t$ denotes the concatenated states of all agents at $t$;
$x^i$ denotes the trajectory of agent $i$ over all stages;
and $x$ denotes the joint trajectory of all agents over the entire horizon $t \!\in\! \{0, \ldots,T\}$.
We use $x^{-i}$ to denote the collection of state variables of all agents other than $i$.
Similar notation applies to control $u$ and Lagrangian multipliers introduced later.

\subsection{Generalized Nash Equilibrium}

The Generalized Nash Equilibrium Problem (GNEP) is formulated as a coupled optimization problem, where agent $i$ performs the following minimization:

\vspace{-0.2in}
\begin{subequations}\label{eq:game}
\begin{align}
\min_{x^i,\,u^i} \quad 
& J^i(x, u^i)=
\sum_{t=0}^{T} J^i_t(x_t, u^i_t) \label{eq:game_obj}\\
\textrm{s.t.} \quad 
& x^i_{t+1} = f(x^i_t, u^i_t), 
&& \hspace{-0.3in} t = 0, \dots, T-1, \label{eq:dyn_con}\\
& h(x, u) \le 0.  \label{eq:ineq_con}
%& h(x_k, u_k) \le 0, 
%&& \hspace{-0.3in} k = 1, \dots, T. \label{eq:ineq_con}
\end{align}
\end{subequations}
Here, $f(\cdot)$ denotes the agent dynamics, $h(\cdot)$ represents shared state and control constraints (e.g., collision avoidance, feasible region), and $J^i_t(\cdot)$ is agent $i$'s running cost.%
\footnote[2]{For exposition simplicity, we consider identical dynamics and constraints, although our analysis extends naturally to heterogeneous settings.}

In this work, we focus on open-loop control strategies, i.e., $u^i_t = u^i(x_0, t)$, where the control inputs depend only on the initial state $x_0$ and the stage $t$. 
When implemented on a physical platform, feedback can be obtained through MPC-style execution of the open-loop control. 

We seek a Generalized Nash Equilibrium (GNE), as it characterizes a broad class of multi-agent interactions ranging from fully cooperative to fully competitive settings.

\begin{definition}
    A joint control strategy $u^* \!=\! (u^{1*}, \dots, u^{N*})$ is a local Generalized Nash Equilibrium if, for every agent~$i$, there exists a \emph{feasible} neighborhood $\mathcal{N}^i$ of $u^{i*}$ such that

    \vspace{-0.2in}
    \begin{equation*}
    % \label{eq:gne}
    J^i(x(u^{i*}, u^{-i*}), u^{i*}) \leq J^i(x(u^{i}, u^{-i*}), u^{i}),  \forall  u^i \in \mathcal{N}^i(u^{i*}).
    \end{equation*}
    % for all $u^i \in \mathcal{N}^i$.
    %$\cap \Omega^i(u^{-i*})$, where $\Omega^i(u^{-i*})$ is the set of feasible controls for agent $i$ given the fixed strategies of other agents $u^{-i*}$.
\end{definition}

The Generalized Nash Equilibrium extends the classical Nash Equilibrium to settings with constraints. 
At a GNE, no agent has an incentive to unilaterally deviate from its equilibrium strategy $u^{i*}$; that is, with $u^{-i*}$ fixed, agent $i$ cannot improve its performance by selecting an alternative control without violating the constraints. 
In general, computing GNE strategies $u^*$ is challenging because it requires solving $N$ coupled, constrained optimal control problems.

%%%%%%%%%%%%%%%%%%%%%%%%%%%%%%%%%%%%%%%%%%%%%%%%%
\section{Spectrum of Multi-Agent Interactions}
%%%%%%%%%%%%%%%%%%%%%%%%%%%%%%%%%%%%%%%%%%%%%%%%%

A GNE is commonly characterized by the first-order optimality (KKT) conditions of each agent's constrained optimization problem. 
For agent $i$, define the Lagrangian of agent $i$ associated with \eqref{eq:game} as

\vspace{-0.2in}
\begin{align}\label{eq:lagrangian}
    \mathcal{L}^i (x,u,\lambda^i, \mu^i) 
    &= \sum_{t=0}^{T} J^i_t(x_t, u^i_t)
    + \mu^\top h(x, u)\nonumber \\ %\sum_{k=1}^{T} (\mu^i_k)^\top h(x, u)
    &\quad + \sum_{t=0}^{T-1} (\lambda^{i}_t)^\top \big[f(x_t^i, u_t^i) - x_{t+1}^i\big],
\end{align}
where $\mu$ and $\lambda^i$ denote the dual variables associated with the inequality constraints and the dynamics, respectively. 
Let $\mathcal{F}^i(x^i,u^i)=0$ represent the vectorized dynamics constraints across stages in~\eqref{eq:dyn_con}.

Then, the KKT conditions for agent $i$ are given by
\begin{subequations}\label{eq:kkt}
\begin{align}
\nabla_{x^i}\mathcal{L}^i &= 0, \quad \nabla_{u^i}\mathcal{L}^i = 0, \label{eq:kkt_stationarity}\\
\mathcal{F}^i(x^i,u^i) &= 0, \label{eq:kkt_dynamics}\\
h(x, u) + s&= 0, \quad s \ge 0,\label{eq:kkt_ineq}\\
\mu^\top h(x,u) &= 0, \quad \mu \ge 0, \label{eq:kkt_comp}
\end{align}
\end{subequations}
where $s$ denotes the slack variable for the inequality constraints. 
Let the stacked variable be $z = (z^1, \dots, z^N)$, where $z^i$ collects the primal variables (i.e., $x^i$ and $u^i$) and dual variables (i.e., $\lambda^i$ and $\mu$). 
Denote the inequality constraint residual $h(x,u)$ as $\mathcal{H}$. 
We define the KKT residual vector for agent $i$ as 
\begin{equation}
    \label{eqn:residue}
    \mathcal{R}^i = (\nabla_{x^i}\mathcal{L}^i, \nabla_{u^i}\mathcal{L}^i, \mathcal{F}^i, \mathcal{H}),
\end{equation}
and the joint residual as $\mathcal{R} = (\mathcal{R}^1, \dots, \mathcal{R}^N)$. 
The KKT condition of GNE is then given by
\[
\mathcal{R}^i(z^1, \dots, z^N) = 0, \qquad i = 1,\dots,N,
\]
or, collectively, as
\begin{equation}
    \label{eqn:master-root-equation}
    \mathcal{R}(z) =
    \left[\mathcal{R}^1(z), \ldots, \mathcal{R}^N(z) \right]^\top
    = 0.
\end{equation}

The interaction structure of the game is revealed by the Jacobian of the KKT operator,
\begin{equation}
    \label{eqn:KKT-matrix}
    \nabla \mathcal{R}(z) =
    \begin{bmatrix}
        \nabla_{z^1} \mathcal{R}^1 & \nabla_{z^2} \mathcal{R}^1 & \cdots & \nabla_{z^N} \mathcal{R}^1 \\
        \nabla_{z^1} \mathcal{R}^2 & \nabla_{z^2} \mathcal{R}^2 & \cdots & \nabla_{z^N} \mathcal{R}^2 \\
        \vdots & \vdots & \ddots & \vdots \\
        \nabla_{z^1} \mathcal{R}^N & \nabla_{z^2} \mathcal{R}^N & \cdots & \nabla_{z^N} \mathcal{R}^N
    \end{bmatrix}.
\end{equation}

The diagonal blocks $\nabla_{z^i} \mathcal{R}^i$ describe the agent $i$’s own influence on its performance, 
while the off-diagonal blocks $\nabla_{z^j} \mathcal{R}^i$ ($j \neq i$) quantify how agent $i$’s optimality conditions respond to variations in agent $j$’s decision. 
In this sense, $\nabla \mathcal{R}(z)$ encodes how each agent must adjust its strategy to improve performance given the actions of others.
%
% \edit{
% Additionally, $\nabla\mathcal{R}(z)$ is always square: its rows correspond to the primal Lagrangian residues and the constraints, while its columns correspond to the primal variables and the dual variables associated with those constraints.
% }
% \panos{$\nabla \mathcal{R}$ may not be square matrix.}

We decompose $\nabla \mathcal{R}(z)$ into its symmetric and skew-symmetric components
\begin{equation}
    \label{eqn:operator-splitting}
    \nabla \mathcal{R}(z) = S(z) + A(z),
\end{equation}
where,
\begin{align*}
S(z) \!&=\! \tfrac{1}{2}\big(\nabla \mathcal{R}(z) \!+\! \nabla \mathcal{R}(z)^{\!\top}\!\big),
A(z) \!=\! \tfrac{1}{2}\big(\nabla \mathcal{R}(z) \!-\! \nabla \mathcal{R}(z)^{\!\top} \!\big).
\end{align*}
This decomposition reveals a structural spectrum of multi-agent games.
The matrix $S(z)$ represents the \emph{aligned objectives} of the game, while matrix $A(z)$ captures the \emph{adversarial interaction}, as illustrated in the following two extreme cases of the games.

\paragraph{Potential Games}

Consider a canonical example, where there exists a scalar \emph{potential} function $\Phi(z)$ such that 
\begin{equation}
    \label{eqn:fully-potential}
    \mathcal{R}(z) = \nabla \Phi(z).
\end{equation}
Then, $\nabla \mathcal{R}(z) = \nabla^2 \Phi(z)$, which is symmetric by construction, and hence $A(z) = 0$.
In particular, $\nabla_{z^i}\mathcal{R}^j = \nabla_{z^j}\mathcal{R}^i = \nabla^2_{z^i z^j}\Phi = \nabla^2_{z^j z^i}\Phi$, which implies that the marginal impact of agent~$j$ on agent~$i$'s performance is identical to that of agent~$i$ on agent~$j$.
Moreover, solving the KKT condition for GNE, i.e., $\mathcal{R}(z) = 0$, can be viewed as seeking a stationary point of the potential function $\Phi(z)$.%
\footnote{It should be noted that our terminology is somewhat different from that used in the game theory literature, which refers to potential games as those imposing a global potential over agents' unilateral deviations: $\Phi(u^i, \!u^{-i}) \!-\! \Phi(u^{i'}, \!u^{-i}) \!= J^{i}(u^{i},\! u^{-i}) \!-\! J^{i}(u^{i'}, \!u^{-i})$. 
Here, we adopt the concept from a local-gradient perspective on the stacked primal-dual variables.}
Thus, in the purely potential case, the equilibrium computation reduces, at least locally, to solving the first-order optimality condition of a single scalar objective.

% \panos{Is this the definition of a potential game?}
% \sg{It's complicated. Our notion of ``potential game" is defined based on the stacked variables (both primal and dual), while the standard potential game is defined based on the ``decision variables / strategies" of each agent only. 
% %
% Also, we defined the potential game using local gradient, but the standard potential game is defined globally. 
% %
% Specifically, let $\sigma^i$ denote an agent's strategy, then, there exist a potential function $\Phi$ such that for all agent $i$ and strategies $\sigma^{i}$ and $\sigma^{i'}$
% \[
%     \Phi(\sigma^i, \sigma^{-i}) - \Phi(\sigma^{i'}, \sigma^{-i}) = J^{i}(\sigma^{i}, \sigma^{-i}) - J^{i}(\sigma^{i'}, \sigma^{-i})
% \]
% However, the relation~\eqref{eqn:fully-potential} resembles the idea of a potential game, where $\Phi$ serves as a potential function.
% }

\paragraph{Competitive Games}

As a canonical example, consider a constraint-free zero-sum game where the cost is $J(z) = z_1^\top G z_2$, 
such that $J^1(z) = J(z)$ and $J^2(z) = -J(z)$. 
Then, 
\[
\nabla \mathcal{R}(z) =
% \begin{bmatrix}
% \nabla_{z^1} \mathcal{R}^1 & \nabla_{z^2} \mathcal{R}^1 \\
% \nabla_{z^1} \mathcal{R}^2 & \nabla_{z^2} \mathcal{R}^2
% \end{bmatrix}
% =
\begin{bmatrix}
\nabla^2_{z^1 z^1} J & \nabla^2_{z^1 z^2} J \\
-\nabla^2_{z^1 z^2} J & -\nabla^2_{z^2 z^2} J
\end{bmatrix}
=
\begin{bmatrix}
0 & G \\
-G^\top & 0
\end{bmatrix}
.
\]
The symmetric component therefore vanishes, i.e., $S(z)=0$, implying that the interaction is purely adversarial.
The GNEP in~\eqref{eq:game} can then be written as a min-max optimization.

% For simplicity, we ignore constraints and consider a state-dependent reward satisfying $J(z) = J^1(z) = -J^2(z)$. 
% Then,
% \[
% \nabla \mathcal{R}(z) =
% \begin{bmatrix}
% \nabla_{z^1} \mathcal{R}^1 & \nabla_{z^2} \mathcal{R}^1 \\
% \nabla_{z^1} \mathcal{R}^2 & \nabla_{z^2} \mathcal{R}^2
% \end{bmatrix}
% =
% \begin{bmatrix}
% \nabla^2_{z^1 z^1} J & \nabla^2_{z^1 z^2} J \\
% -\nabla^2_{z^1 z^2} J & -\nabla^2_{z^2 z^2} J
% \end{bmatrix}.
% \]
% The symmetric component therefore vanishes, i.e., $S(z)=0$, implying that the interaction is purely adversarial.
% The GNEP in~\eqref{eq:game} can then be written as a min-max optimization. 

%%%%%%%%%%%%%%%%%%%%%%%%%%%%%%%%%%%%%%%%%%%%%%%%%
\section{Potential-Dominant Games}
%%%%%%%%%%%%%%%%%%%%%%%%%%%%%%%%%%%%%%%%%%%%%%%%%

Most real-world multi-agent systems lie between the extremes of purely potential games and fully competitive games.
% %
% For example, in warehouse robotics and traffic coordination, agents often share broad system-level objectives, such as improving throughput or maintaining safety.
% %
% At the same time, their local objectives may not be fully aligned; for instance, a merging vehicle may prefer to enter ahead of another vehicle.
% %
% Nevertheless, the cooperative objective typically dominates the optimization landscape, with competitive effects acting as perturbations.
%
This motivates the following definition of $\gamma$-Potential-Dominant Games ($\gamma$-PDGs).

\begin{definition}[$\gamma$-Potential-Dominant Games]
For a given $\gamma \in (0,1]$, a game is said to be $\gamma$-potential-dominant at $z$ if $S(z)$ is nonsingular and
\[
\rho\left(S(z)^{-1}A(z)\right)<\gamma,
\]
where $\rho(\cdot)$ denotes the spectral radius.
\end{definition}

The above definition of $\gamma$-potential-dominant games formalizes the idea that potential structure dominates adversarial interactions.
When $\rho(S^{-1}A) = 0$, the game reduces to a strict potential game.
As $\rho(S^{-1}A)$ increases, the influence of non-cooperative interactions strengthens, and the game progressively departs from aligned interests toward strongly competitive behaviors.

Although $\gamma$-potential-dominant games share a similar motivation with near-potential games~\cite{candogan2011flows} and Markov $\alpha$-potential games~\cite{guo2025markov, kalaria2024alpha}, our formulation differs in two important aspects.
\textit{First}, $\alpha$-potential games measure deviation from a potential game through a global scalar approximation error associated with a surrogate potential function, whereas $\gamma$-potential dominance is characterized locally through the spectral properties of the game KKT matrix.
\textit{Second}, existing $\alpha$-potential formulations typically optimize the surrogate potential and therefore recover an approximate equilibrium of the original game.
In contrast, our method explicitly retains the competitive component and incorporates it through the iterative refinement procedure introduced later, thereby converging to the exact local GNE under the potential-dominance condition.

\subsection{Verification of Potential Dominance}

We first establish two equivalent characterizations of potential dominance in terms of a Lyapunov inequality and operator norm.
The proof is presented in the appendix.

\begin{theorem}
    \label{thm:equivalent-lyapunov}
    The following conditions are equivalent:
    \begin{enumerate}
        \item[(i)] The game is $\gamma$-potential-dominant at $z$.
        \item[(ii)] There exists a positive-definite matrix $P$ such that 
            \begin{equation}
                \label{eqn:equiv-potential-dom-cond}
                A(z)^\top P A(z) - \gamma^2 S(z)^\top P S(z) \prec 0.
            \end{equation}
        \item[(iii)] There exists a positive-definite matrix $Q \succ 0$ such that 
            \( \norm{Q}{S^{-1}(z) A(z)} < \gamma,\)
        where the operator norm is induced by the vector norm $\|x \|_Q^2 = x^\top Q x.$
    \end{enumerate}

\end{theorem}
\vspace{+10pt}

Although Theorem~\ref{thm:equivalent-lyapunov} provides an exact characterization, verifying it requires solving a semidefinite program, which may be prohibitively expensive.
To obtain a tractable sufficient condition, we sacrifice necessity by fixing $P=I$, thereby restricting the analysis to the standard Euclidean metric and leading to the following corollary.

\begin{corollary}
    \label{cor:cholesky-test}
    If $\norm{2}{S^{-1}(z)A(z)} = \norm{I}{S^{-1}(z)A(z)} < 1$, then the game is $\gamma$-potential-dominant at $z$ for any $\gamma \in \left( \Vert{S^{-1}(z)A(z)}\Vert_2, 1 \right]$.
\end{corollary}
% \vspace{+10pt}

Corollary~\ref{cor:cholesky-test} reduces verification of potential dominance to estimating the largest singular value of $S(z)^{-1}A(z)$.
This quantity can be computed using Krylov-subspace methods that require only matrix-vector products.
In particular, multiplications involving $S(z)^{-1}$ need not be formed explicitly and can instead be evaluated through forward and backward substitutions using a precomputed \ldlt{} factorization of $S(z)$.
The resulting test, however, is conservative: failure of condition in Corollary~\ref{cor:cholesky-test} does not preclude the game from being $\gamma$-potential-dominant under a different matrix $P\neq I$.

% \subsection{An Interpretation of Potential Dominance}

% Corollary~\ref{cor:cholesky-test} provides an intuitive interpretation of potential dominance.
% %
% For simplicity, consider an $N$-agent, single-stage, unconstrained game with scalar actions $u^i \! \in\!  \mathbb{R}$, where agent $i$ minimizes $J^i(u^1,\!\dots,\!u^N)$.
% %
% In this setting,
% $[\nabla \mathcal{R}(z)]_{ij} = \nabla^2_{u^i u^j} J^i$.

% We define the \emph{pairwise adversarial coupling} between agents $i$ and $j$ as
% $$C_{ij}  = \frac{1}{2} \left| \nabla_{u^i u^j}^2 (J^i - J^j) \right| = |A_{ij}|,$$
% and the corresponding \emph{aggregate adversarial influence} on agent $i$ as $c_i=\sum_{j\neq i}c_{ij}$.

% Conversely, we define the \emph{cooperative stability margin} of agent $i$ as the diagonal-dominance margin of $S$:
% $$D_i  = \nabla_{u^i u^i}^2 J^i - \sum_{j \neq i} \frac{1}{2} \left| \nabla_{u^i u^j}^2 (J^i + J^j) \right| = S_{ii} - \sum_{j \neq i} |S_{ij}|$$

% This yields the following sufficient condition.

% \begin{corollary}
% \label{cor:CD-suff-con}
% A game is $\gamma$-potential-dominant if
% $$\max_i C_i \;\leq\; \gamma \min_i D_i.$$
% \end{corollary}

% The result follows directly from Corollary~\ref{cor:cholesky-test}, using the following inequality
% \[
%     \|S^{-1} A\|_2 \leq \|S^{-1}\|_2 \|A\|_2 \leq \frac{\max_i C_i}{\min_i D_i}.
% \]
% Thus, potential dominance admits a natural interpretation: the weakest cooperative stability margin must outweigh the strongest aggregate adversarial influence.

%%%%%%%%%%%%%%%%%%%%%%%%%%%%%%%%%%%%%%%%%%%%%%%%%
\section{\texorpdfstring{A Fast Solver for $\gamma$-PDGs via Symmetric Splitting and Iterative Refinement}{A Fast Solver for gamma-PDGs via Symmetric Splitting and Iterative Refinement}}
%%%%%%%%%%%%%%%%%%%%%%%%%%%%%%%%%%%%%%%%%%%%%%%%%

For general $N$-agent games, existing numerical approaches~\cite{algames, rd3g} typically compute a GNE by applying gradient-based root-finding methods to the KKT system~\eqref{eqn:master-root-equation}.
For instance, a standard Newton-type update rule takes the form
\[
    z_{k+1} = z_k + \Delta z^*_k,
\]
where the descent step $\Delta z^*_k$ is obtained by solving 
\begin{equation}
    \label{eqn:root}
    \nabla \mathcal{R}(z_k)\, \Delta z^*_k = -\mathcal{R}(z_k).
\end{equation}
    
\subsection{Symmetric Splitting}
In multi-agent settings, however, the Jacobian $\nabla \mathcal{R}(z_k)$ can be both large and ill-conditioned, particularly as the number of agents increases or the coupling constraints become more complex.
Consequently, repeatedly solving this linear system can be computationally expensive and slow, limiting the applicability of such methods for real-time MPC.

For potential-dominant games, we can exploit the operator splitting in~\eqref{eqn:operator-splitting}, namely,
$\nabla \mathcal{R}(z_k)=S(z_k)+A(z_k)$,
where the symmetric $S(z_k)$ captures aligned interactions and the skew-symmetric $A(z_k)$ captures adversarial coupling.
In the remainder of this section, we write $S$ and $A$ for $S(z_k)$ and $A(z_k)$ whenever the evaluation point is clear.

\subsection{Preconditioning and Iterative Refinement}
Rather than directly solving
$
(S+A)\Delta z
=
-\mathcal{R}(z_k),
$
we use the potential component $S$ to precondition the system.
Left-multiplying by $S^{-1}$ gives
\[
\left(I+S^{-1}A\right)\Delta z
=
-S^{-1}\mathcal{R}(z_k).
\]
This transformation isolates the identity term and treats the competitive component $S^{-1}A$ as a perturbation.
In implementation, we do not explicitly form $S^{-1}$.
Instead, we apply the following \emph{iterative refinement scheme} to find $\Delta z_k^*$:
\begin{equation}
\label{eq:iter}
S(z_k) \Delta z_k^{(j+1)}
=
-\mathcal{R}(z_k)-A(z_k) \Delta z_k^{(j)}.
\end{equation}
Thus, each refinement step requires only a linear solve with $S$, rather than with the full Jacobian $S+A$.

Given the current evaluation point $z_k$, since $S(z_k)$ is symmetric, \eqref{eq:iter} can be solved efficiently using an \ldlt{} factorization.
Because $S(z_k)$ remains fixed throughout the iterative refinement (indexed by $j$), its factorization is computed only once and reused at every refinement step, which then requires only inexpensive triangular solves~\cite{trefethen2022numerical}.
Moreover, when the sparsity pattern of $S$ remains unchanged across MPC iterations (indexed by $k$), symbolic factorization can also be reused, leaving only the numerical values to be updated, thus speeding up computations significantly.

\begin{algorithm}[t]
\caption{Potential Dominant Operator Splitting}
\label{alg:algorithm}
\begin{algorithmic}
\REQUIRE Initial guess $z$, Residue tol $\tau$, Refinement tol $\varepsilon$
\REPEAT
    \STATE \textbf{Compute:} $\mathcal{R}(z) \text{~and~} \nabla \mathcal{R}(z)$ per~\eqref{eqn:residue} and \eqref{eqn:KKT-matrix}
    \STATE \textbf{Decomposition} $S = \frac12(\nabla \mathcal{R} + (\nabla \mathcal{R})^\top), A = \nabla \mathcal{R} - S$
    \STATE \textbf{Factorize:}  $S = L D L^T$
    \STATE \textbf{Initialize refinement:} $\Delta z^{(0)} \leftarrow 0$
    \STATE \textbf{Iterative Refinement:}
    \REPEAT
        \STATE Solve $S \Delta z^{(j+1)} = - \mathcal{R} - A \Delta z^{(j)}$ using \ldlt{}
    \UNTIL $\Vert \Delta z^{(j+1)} - \Delta z^{(j)}\Vert  < \frac{1-\gamma}{\gamma}\varepsilon$
    \STATE \textbf{Update:} $z \leftarrow z + \Delta z^{(j+1)}$
\UNTIL{$\vert \mathcal{R}(z)\vert < \tau$}
\RETURN $z$
\end{algorithmic}
\end{algorithm}

\subsection{Convergence and Algorithm Complexity}

Next, we show that the $\gamma$-potential-dominance ensures that the refinement scheme in~\eqref{eq:iter} is contractive, thereby guaranteeing convergence to the exact Newton step. 

\begin{theorem}
Suppose the game is $\gamma$-potential-dominant at $z_k$.
Then, the sequence ${\Delta z_k^{(j)}}$ generated by~\eqref{eq:iter}
converges \emph{linearly} to the solution $\Delta z^*_k$ of~\eqref{eqn:root}.
In particular, letting $e^{(j)}=\Delta z_k^{(j)}-\Delta z^*_k$.
There exists a positive-definite matrix $P$ such that
\[
\Vert e^{(j)} \Vert_P
\leq
\gamma^j \Vert e^{(0)}\Vert_P,
\]
where $\norm{P}{x}^2 = x^\top P x$.
Additionally, the iteration complexity is given by $\mathcal{O} \Big(\frac{\log (1/\varepsilon)}{\log(1/ \gamma)} \Big)$.
\end{theorem}

\begin{IEEEproof}
    Subtracting $S \Delta z_{k}^* = - \mathcal{R}(z_k) - A \Delta z^*_k$ from~\eqref{eq:iter} yields
    \(
        S(\Delta z^{(j+1)}_k - \Delta z_k^*) = -A(\Delta z_k^{(j)} - \Delta z^*_k).
    \)
    Hence, the error evolves according to
    \[
        e^{(j+1)} = -S^{-1}A\, e^{(j)}.
    \]
    Since the game is $\gamma$-potential-dominant, Theorem~\ref{thm:equivalent-lyapunov} guarantees the existence of a positive-definite matrix $P$ such that 
    such that the matrix norm satisfies $\norm{P}{S^{-1}A} < \gamma$.
    The desired bound and complexity follow directly:
    \[
        \norm{P}{e^{(j+1)}}
        \leq
        \norm{P}{-S^{-1}A} \norm{P}{e^{(j)}}
        <
        \gamma\norm{P}{e^{(j)}}.
    \]

    Therefore, to achieve a desired error $\norm{P}{e^{(j)}}\leq\varepsilon$, it suffices to take
    $j \geq \frac{\log( \varepsilon /\Vert e^{(0)}\Vert_P)}{\log(\gamma)} = \frac{\log(\Vert e^{(0)}\Vert_P /  \varepsilon)}{\log(1/\gamma)}$.
    Thus, the iterative refinement scheme converges linearly, with iteration complexity $\mathcal{O} \Big(\frac{\log (1/\varepsilon)}{\log(1/ \gamma)} \Big)$. 
\end{IEEEproof}

The above theorem also provides an a posteriori bound on the distance between the final iterate and the exact solution when the refinement is terminated early.

\begin{corollary}
Suppose the iteration in~\eqref{eq:iter} is terminated at step $j$ and satisfies 
\(
\norm{P}{\Delta z_k^{(j)}-\Delta z_k^{(j-1)}} 
\leq \varepsilon.
\)
Then the final iterate satisfies
\[
\norm{P}{\Delta z_k^{(j)}-\Delta z_k^*}
\leq \frac{\gamma}{1-\gamma} \varepsilon.
\]
\end{corollary}

\section{Numerical Simulations}
%=============================================%

We evaluate the proposed approach on six multi-agent dynamic-game scenarios.
We provide detailed numerical analyses and discussions for two representative problems, highway merging and autonomous racing, while illustrating agents' behaviors in the other four scenarios through figures.
Across these experiments, we compare the runtime of the proposed method against a baseline solver that directly solves the generalized Nash equilibrium without symmetric splitting and iterative refinement.

\subsection{Implementation Details}
The dynamics, cost functions, and constraints are defined in Python using CasADi~\cite{casadi}.
The core algorithm is implemented in C++ to maximize execution speed, with algorithmic differentiation for the game KKT provided by CasADi.
All simulations employ a prediction horizon of $T=20$ steps.
The tolerance for the KKT residual is set to $\tau = 5 \times 10^{-4}$ for all trials.
Benchmarking experiments were executed on a single thread of a desktop computer equipped with an Intel i7-7700k processor with 32\,GB of memory.

\subsection{Highway Merging Scenario}

\begin{figure}[!b]
  \centering
  \begin{subfigure}{\linewidth}
  \centering
  \includegraphics[clip, trim=300 52 280 52,angle=-90,width=0.95\linewidth]{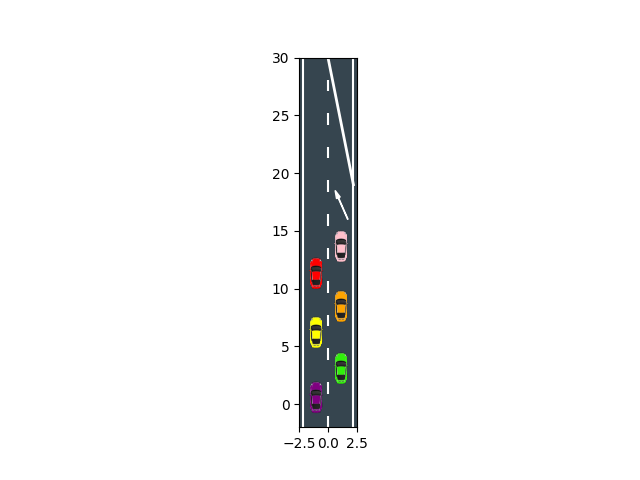}
    \vspace{-5pt}
    \caption{Initial State ($k=0$)}
  \end{subfigure}
  \begin{subfigure}{\linewidth}
  \centering
    \vspace{5pt}
  \includegraphics[clip, trim=300 52 280 52,angle=-90,width=0.95\linewidth]{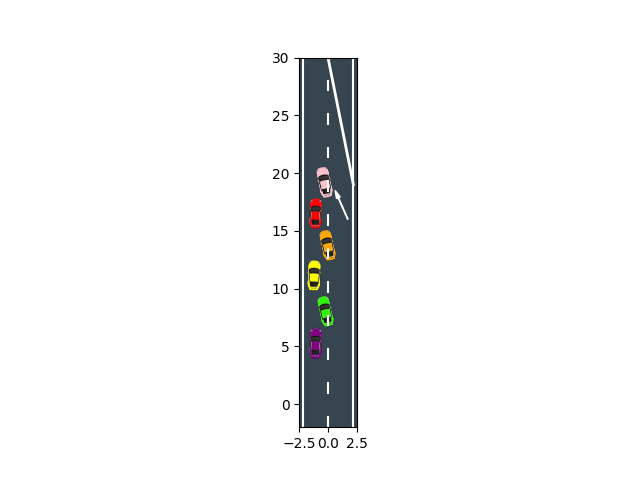}
  \vspace{-5pt}
    \caption{Intermediate State ($k=10$)}
  \end{subfigure}
  \begin{subfigure}{\linewidth}
  \centering
  \vspace{5pt}
  \includegraphics[clip, trim=300 52 280 52,angle=-90,width=0.95\linewidth]{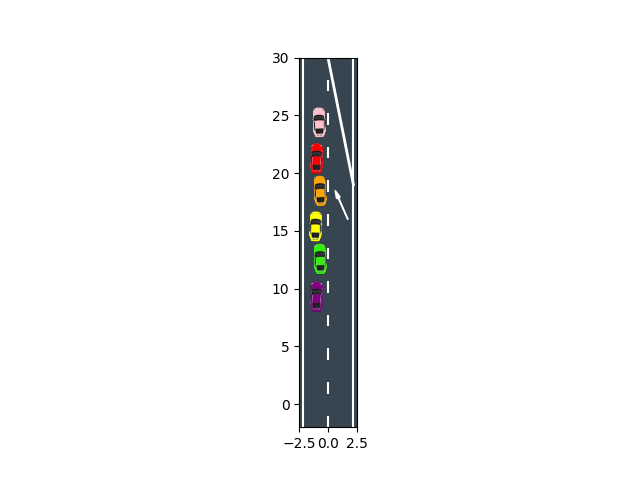}
  \vspace{-5pt}
    \caption{Final State ($k=20$)}
  \end{subfigure}
  \vspace{-20pt}
  \caption{Time-lapse of the 6-car merging game. 
  Agents successfully merge while maintaining safety distances.
  }
\label{fig:merge}   
\end{figure}

The highway merging scenario serves as a standard benchmark for evaluating mixed cooperative and competitive multi-agent behaviors.
Vehicles are initialized with randomized positions, headings, and speeds across two adjacent lanes.
The objective is to merge into the left lane without collisions.

Vehicle dynamics are governed by a standard kinematic bicycle model~\cite{rajamani2006vehicle}, with the states defined as $[x, y, v, \phi]^\top$, where $ x$ and $ y$ denote the longitudinal and lateral positions, $ v$ the speed, and $ \phi$ the heading.
The stage cost is formulated as a quadratic function penalizing deviations from the target lane, reference speed, heading deviation, and control effort.
Additionally, a small adversarial reward is added to incentivize players to maintain a lead relative to others, rendering the problem a general-sum game.
Collision avoidance is formulated using non-convex inequality constraints that enforce a minimum distance between any two vehicles.

Figure~\ref{fig:merge} shows a solution obtained using PDO-Split in the six-vehicle merging scenario.
Note that the initial gaps between cars are insufficient for a safe merge if all vehicles maintain their current speeds.
The agents therefore coordinate their speeds and lane-change timing to complete the maneuver successfully within the limited merging distance.
The resulting trajectories operate close to the collision constraints, highlighting both the difficulty of the scenario and the solver's ability to identify an efficient yet feasible solution.
Finally, over 100 random initial conditions, the average value of $\gamma$ is 4.3E-8, indicating that inter-agent interactions in the highway-merging scenario are strongly dominated by the potential component.

\subsection{Car Racing Scenario}

\begin{figure}[b!]
\centering
% upper, left, down right
\includegraphics[angle=0,width=\columnwidth]{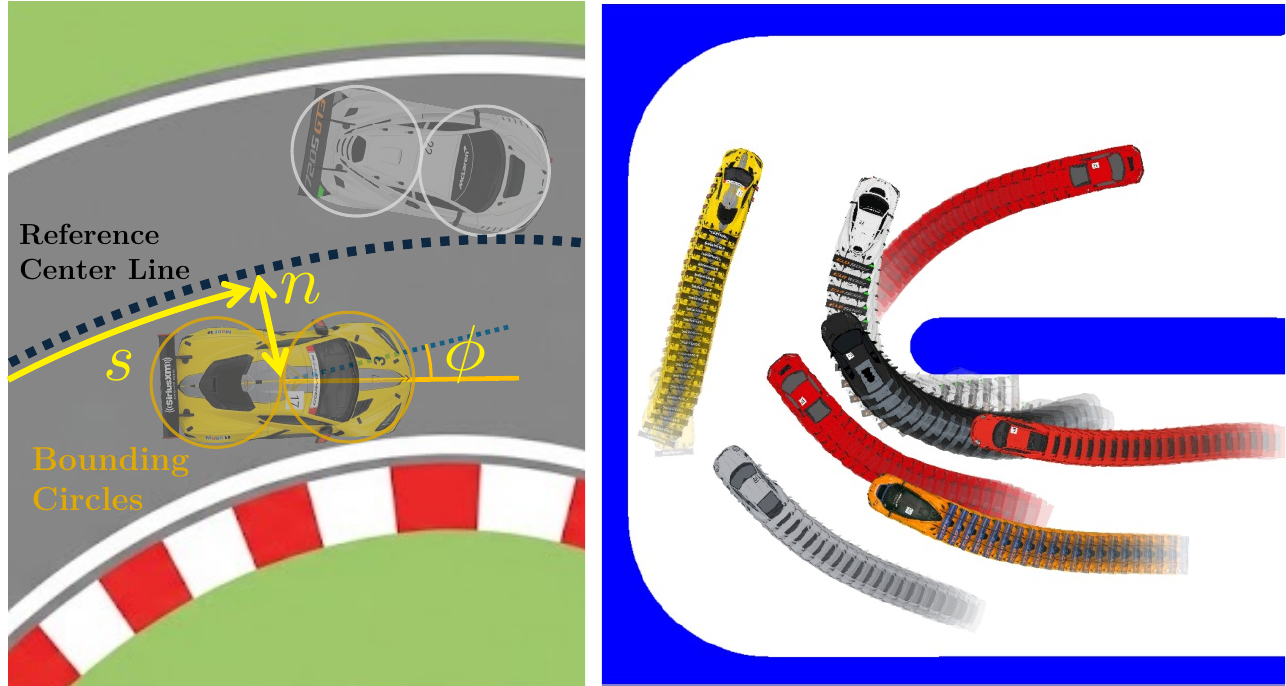}
\caption{\textbf{Left}: An illustration of the racing gamee. 
\textbf{Right}: 
A multi-exposure of the eight-car racing simulation. 
Agents exhibit aggressive overtaking and close-proximity maneuvers. 
}
\label{fig:racing_simulation}
\end{figure}

To evaluate the solver under highly dynamic and more adversarial conditions, we implemented a multi-agent car racing simulation.
The vehicle states are defined in the Frenet frame~\cite{kim2023frenet}, denoted as $ [s, n, \phi, v_f, v_s]^\top$.
Here, $s$ denotes the longitudinal arc length along the reference centerline, $n$ is the lateral deviation, $\phi$ is the heading angle relative to the reference tangent, and $ v_f$ and $ v_s$ are the longitudinal and lateral speeds, respectively.
The agents follow a kinematic bicycle model projected onto this curvilinear space~\cite{rajamani2006vehicle}.
The cost function is formulated to maximize progress along the path while minimizing lateral deviation, heading deviation, and speed deviation from a target speed. 
Safety is enforced through (i) track-boundary constraints on the lateral deviation $n$, and (ii) inter-agent collision constraints, where each vehicle is represented by two bounding circles.
Figure~\ref{fig:racing_simulation} illustrates the Frenet frame and the bounding circles used in the racing game.

% Safety is ensured via two sets of spatial constraints.
% %
% First, strict track boundary limits are imposed on the lateral deviation $n$ to keep vehicles securely within the drivable surface.
% %
% Second, inter-agent collision avoidance is enforced by modeling each car as two distinct bounding circles.
% %
% A collision constraint is violated if a circle associated with one vehicle intersects the circle of another vehicle.

%
% the Euclidean distance between any circle of one car and any circle of an opponent falls below the predefined safety distance.

The average value of $\gamma$ in the racing scenario is 2.2E-4.
This value is two orders of magnitude larger than that of the highway-merging scenario, indicating that the racing game exhibits substantially stronger competitive interactions.

Figure~\ref{fig:racing_simulation} illustrates the trajectories of eight vehicles in the racing game.
Overall, the agents race in close proximity while avoiding collisions and track-boundary violations.
Noticeably, the white car takes an aggressive inside line at the corner entry, forcing the black car and the middle red car to move outward to avoid a collision.
%
%At the corner exit, the white car shifts to a wider line and blocks both vehicles from overtaking. 
%\sg{The white car "shifts to a wider line and blocks both vehicles" is not very clear from Fig 2.}
% while the leading red car, unobstructed by nearby competitors, follows a nominal racing line.
%
Similar emergent behaviors, including blocking, overtaking, and other human-like interactions, are consistently observed in both simulation and physical experiments.
Further details of the physical experiments are provided in Section~\ref{sec:exp}.

\subsection{Benchmark Results}
We benchmarked the computational speed (wall-clock time) of the proposed iterative refinement method against a direct variational GNE solver.
Both methods were tasked with finding a local GNE from identical randomized initial conditions, tracking the runtime necessary to reach the specified residual tolerance.
As illustrated in Fig.~\ref{fig:runtime_benchmarks}, the proposed symmetric splitting approach yields significant computational advantages.
By utilizing the potential component $S$ as a preconditioner, our method effectively decouples the most computationally demanding segments of the Newton step, thus reducing the wall-clock time per iteration while maintaining convergence.

The initial states used in the benchmarks were chosen to place cars closely together to promote dense interaction.
In a more realistic situation, the interaction is usually much sparser, leading to faster and easier convergence.

{
\begin{figure}[tbp]
\centering
\begin{subfigure}{\columnwidth}
\centering
\includegraphics[ width=\textwidth]{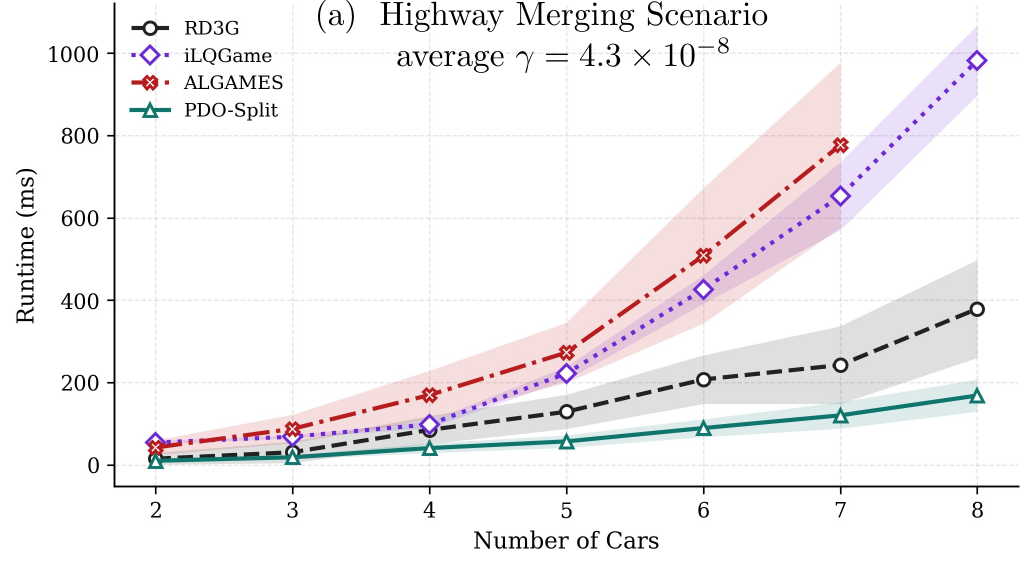}
\vspace{-20pt}
% \caption{Highway Merging Scenario: average $\gamma = 4.3\times 10^{-8}$}
\end{subfigure}\hfill

\vspace{+15pt}
\begin{subfigure}{\columnwidth}
\centering
\includegraphics[width=\textwidth]{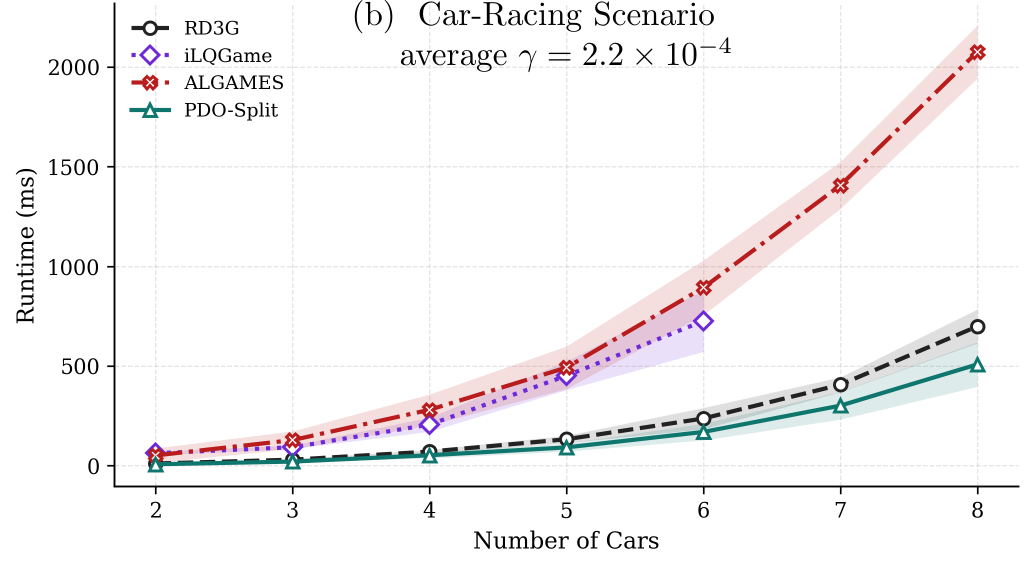}
\vspace{-20pt}
% \caption{Car-Racing Scenario: average $\gamma = 2.2 \times 10^{-4} $}
\end{subfigure}
\caption{Runtime benchmarks comparing the existing solvers with the proposed PDO-Split method, averaged over 100 initial conditions, with shade indicating the standard deviation.
Missing data points indicate failure to return a solution.
}
\label{fig:runtime_benchmarks}
\end{figure}
}

\subsection{Additional Simulations}

\begin{figure*}
    \centering
    % Width is set to 0.24 to leave a small gap (\hfil) between the 4 images
    \subfloat[\shortstack{Air Traffic Control\\ $\gamma = 3.7\!\times\!10^{\text{-}8}$}]{
        \includegraphics[height=2.7cm,clip,trim=100 100 40 130,]{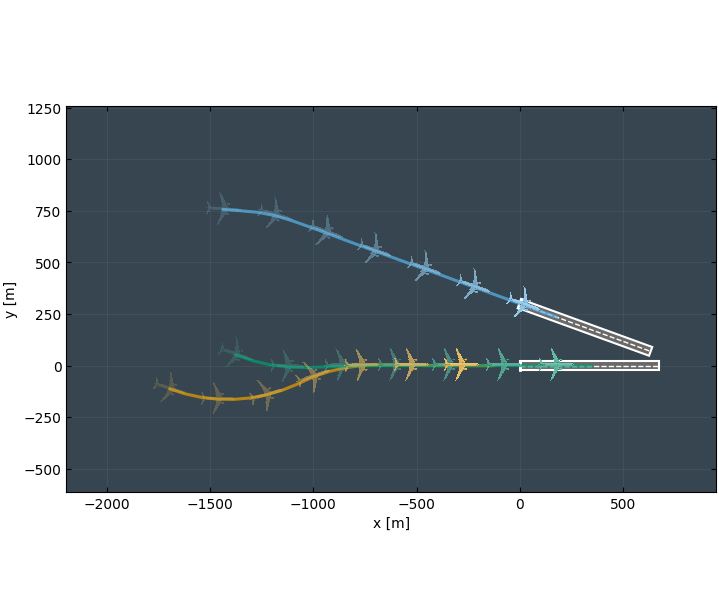}
        \label{fig:atc//}
    }
    \subfloat[\shortstack{Tandem Drifting \\ $\gamma = 2.6 \!\times\!10^{\text{-}7}$}]{
        \includegraphics[height=2.7cm,clip,trim=60 75 50 90,]{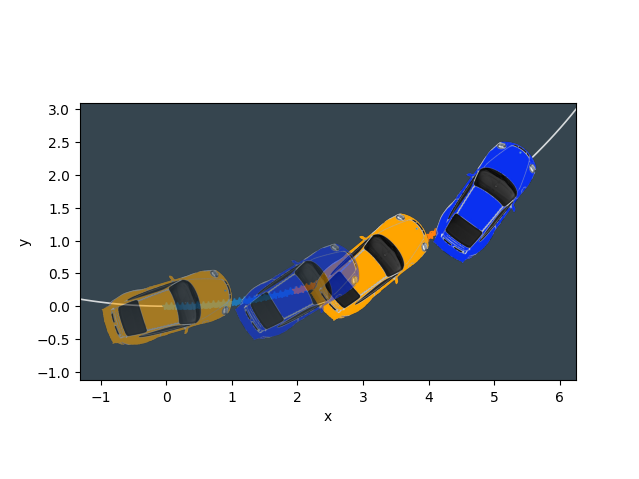}
        \label{fig:drift}
    }
    \subfloat[\hspace{-5pt} \shortstack{Tractor-Trailer\\  ~~$\gamma = 8.1 \!\times\! 10^{\text{-}8}$}]{
        \includegraphics[height=2.7cm,clip,trim=110 50 80 60,]{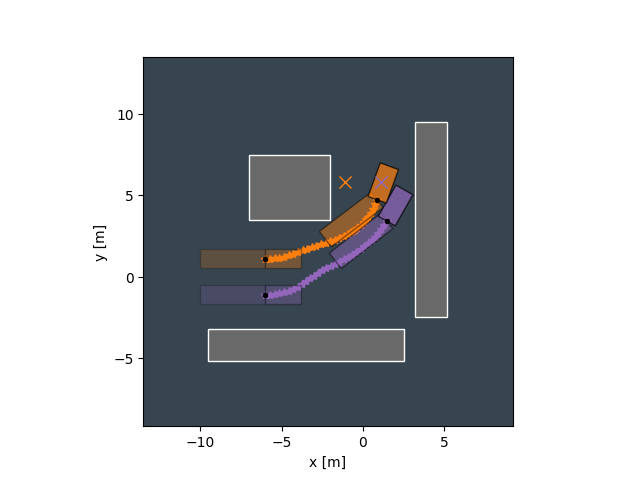}
        \label{fig:trailer}
    }
    \subfloat[\hspace{-8pt}\shortstack{Rocket Landing \\ ~~~$\gamma = 6.9 \!\times\! 10^{\text{-}13}$}]{
        \includegraphics[height=2.7cm,clip,trim=100 45 80 50,]{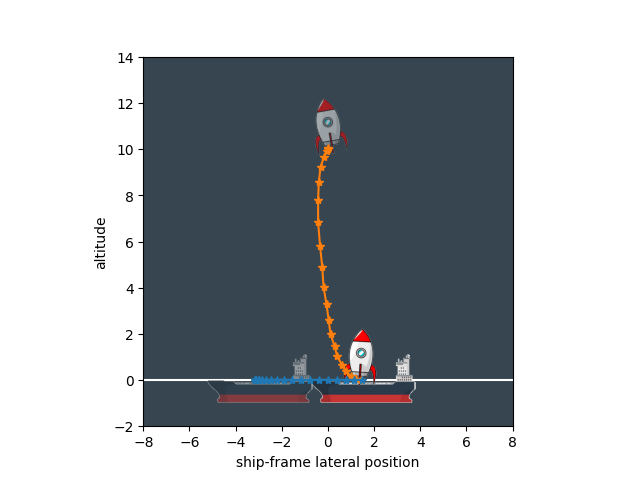}
        \label{fig:rocket}
    }
    \vspace{-10pt}
    \caption{Visualizations of agent behaviors in additional dynamic games.}
    \label{fig:all_games}
\end{figure*}

To further demonstrate the versatility and generalizability of the proposed iterative-refinement solver, we evaluated it on a broader suite of dynamic games.
Visualizations of these additional scenarios are provided in Fig.~\ref{fig:all_games}, together with brief descriptions of their objectives and interaction structures.
The observed computational performance and qualitative behaviors are consistent with those reported for the racing and merging problems; hence, we omit detailed analysis due to space limitations.
Simulation animations are included in the supplemental material.
The source code and simulation details will be made publicly available on the project website upon publication.

% \panos{These simulations are very brief to be of any use to the reader. Can we provide more details? Are you planning to provide the videos?}
% \nick{We discussed two popular games in detail and showed their benchmark results. I believe this is sufficient and consistent with other papers in the field, especially since we also did a physical experiment. 
% The additional examples are there just to show readers that the solver is applicable to a variety of types of problems. 
% I think we can show the simulation snapshots and not go into details, the gifs, videos etc. will be on the project website.
% }
% \panos{"ust to show readers that the solver is applicable to a variety types of problems": This is not shown however! You only say that you tested it and you have a picture that is not informative. I think adding a couple of sentences and adding the videos from the animations would be helpful. Otherwise, }
% \nick{Sounds good. I edited the discussion to explain the unique aspect for each problem, and noted that videos of the simulations will be included in the experiment video.}

\noindent
\textit{a) Decentralized Air Traffic Control:}
Three autonomous fixed-wing aircraft coordinate their approaches to two runways.
Each aircraft seeks to minimize its landing time while maintaining safe separation from the others.
As shown in Fig.~\ref{fig:all_games}(a), the yellow aircraft slows down and takes a small detour to maintain separation from the green aircraft while approaching the horizontal runway.
Meanwhile, the green aircraft accelerates to create sufficient spacing for the yellow aircraft.
%
% In the merging and racing game, the cost functions are convex, and the non-convexity mainly comes from the collision constraints.
%
The multiple runways in this problem lead to a nonconvex cost landscape and constraints for each agent, thus testing the PDO-Split's performance under nonconvexity.

\noindent
\textit{b) Tandem Drifting:}
Two autonomous vehicles perform a coordinated drifting maneuver in close proximity.
In Fig.~\ref{fig:all_games}(b), we observe that each vehicle maintains stability, avoids collision, and synchronizes its motion with the other vehicle throughout the maneuver.
This scenario highlights the ability to operate under highly nonlinear dynamics near physical limits.

\noindent
\textit{c) Tractor--Trailer Maneuvering:}
A tractor-trailer vehicle system is commanded to maneuver in a tightly constrained corridor.
As seen in Fig.~\ref{fig:all_games}(c), the two vehicles successfully avoid collisions with each other, as well as avoid the surrounding obstacles, while respecting the constraints on the maximum misalignment between the tractor and trailer.
%which in practice can cause vehicle damage.
%
Compared to the Ackermann-steered vehicles~\cite{rajamani2006vehicle}, this scenario involves more challenging dynamics, demonstrating the ability to handle non-holonomic motion in obstacle-dense environments.

\noindent
\textit{d) Rocket Landing:}
A reusable booster and an Autonomous Spaceport Drone Ship (ASDS) are required to coordinate a dynamic rendezvous.
The ASDS begins with a random initial velocity and is tightly constrained by acceleration limits. 
Rather than forcing the sluggish ASDS to a complete halt, the solver discovers a coupled, energy-minimizing cooperative strategy (Fig.~\ref{fig:all_games}(d)).
The rocket matches the ship's lateral velocity and intentionally touches down at a slight tilt, riding the boundary of the maximum allowable misalignment constraint.
This example highlights the algorithm's capacity to exploit the full constraint envelope in fully cooperative settings, discovering efficient maneuvers.
%\sg{What's counter intuitive?}

%=============================================%
\section{Experimental Validation}\label{sec:exp}
%=============================================%

\begin{figure}[b!]
\vspace{-10pt}
  \centering
  \begin{subfigure}[b]{0.48\linewidth}
  \includegraphics[width=\linewidth]{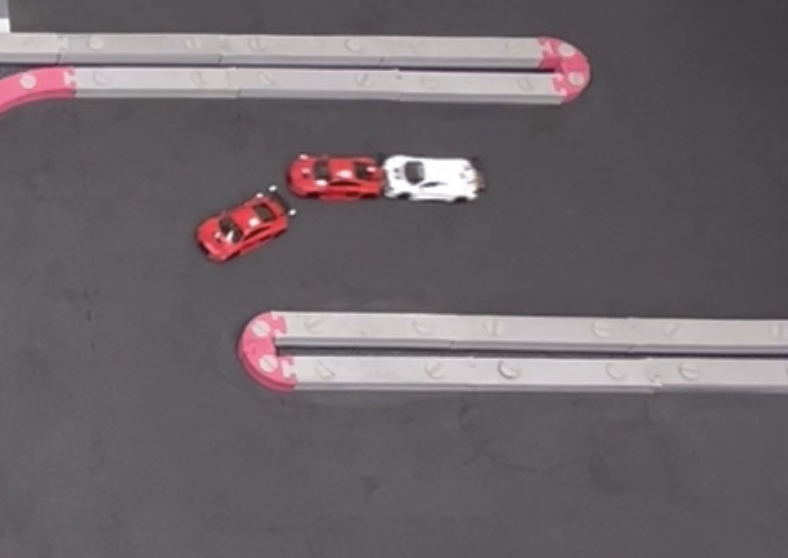}
  \end{subfigure}
  \begin{subfigure}[b]{0.48\linewidth}
  \includegraphics[width=\linewidth]{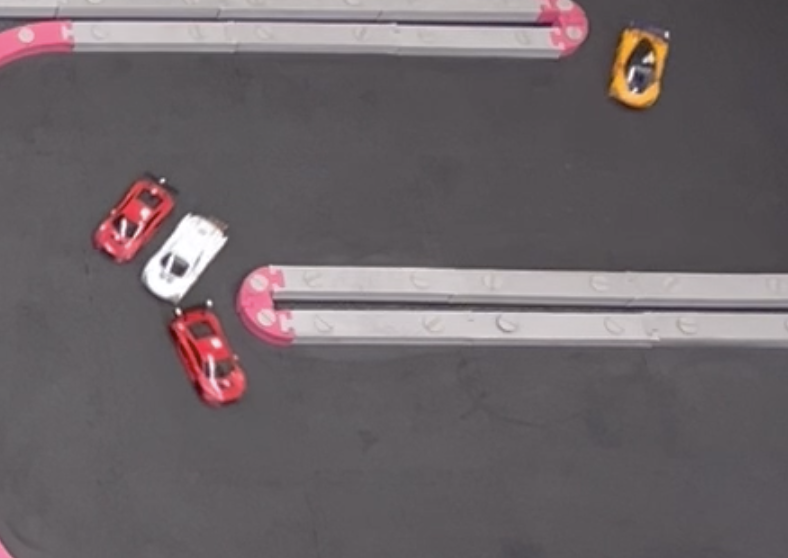}
  \end{subfigure}
  \\[0.8ex]
  \begin{subfigure}[b]{0.48\linewidth}
  \includegraphics[width=\linewidth]{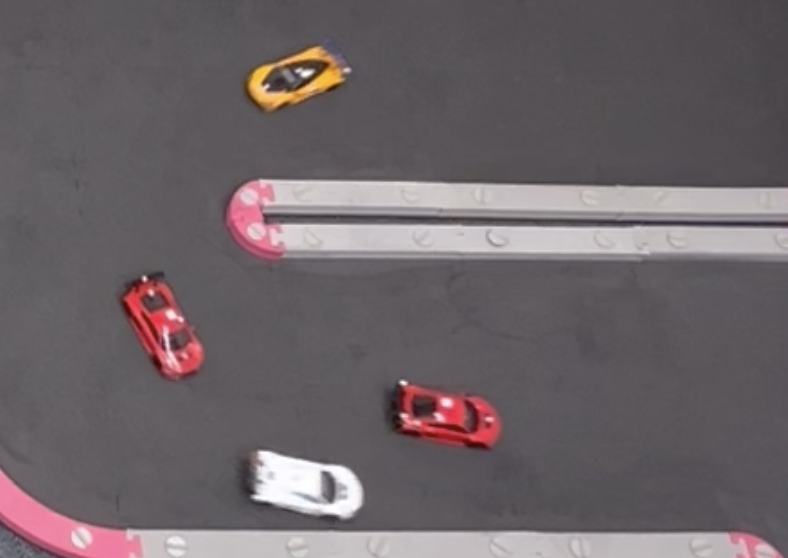}
  \end{subfigure}
  \begin{subfigure}[b]{0.48\linewidth}
  \includegraphics[width=\linewidth]{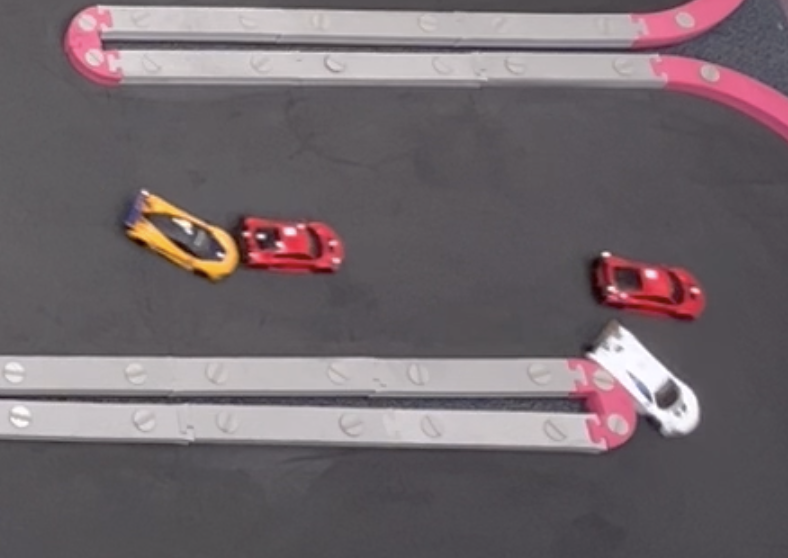}
  \end{subfigure}
  \caption{Snapshots of an overtake during the experiment.
  }\label{fig:exp}
\end{figure}

To validate the robustness and real-time capability of the proposed algorithm in the physical world, we deployed the solver with the car racing problem as a receding horizon planner on a fleet of miniature autonomous race cars.
The experiment utilizes eight BuzzRacer~\cite{buzzracer} platforms navigating a densely populated, tightly constrained closed-circuit track.
Operating eight vehicles simultaneously presents a formidable challenge for real-time game-theoretic planners due to the exponential growth in interaction complexity and the high density of active collision constraints.
The proposed solver calculates a local GNE over the 20-step prediction horizon for all eight agents at every planning cycle.
These high-level reference trajectories are subsequently tracked by decentralized, vehicle-specific lower-level controllers operating at 100Hz.
%Despite the increased dimension of the joint optimization problem compared to previous 4-car iterations, 
The iterative refinement solver consistently yields valid GNE trajectories within the strict timing margins required for continuous, high-speed operation. 

During the physical trials,  the solver consistently finds solutions to the 8-car racing problem in less than 200\,ms, significantly lower than the average runtime reported in the benchmark simulations. 
This difference is due to the tightly coupled initial state we used in the benchmarks.
In a real-world scenario, cars are rarely in such close proximity, making the game easier to solve.
The obtained control strategies exhibited complex adversarial behaviors, such as defensive blocking, coordinated overtakes, and sustained side-by-side racing, without violating track limits or inter-agent safety boundaries.
Figure~\ref{fig:exp} shows snapshots of an overtake (white car) observed during the experiment. 
Video of the physical experiment is provided in supplemental materials and will be made publicly available upon publication.

%=============================================%
\section{Conclusion}
%=============================================%

This work introduced the class of $\gamma$-potential-dominant games and exploited its structure to accelerate game-theoretic model predictive control.
By decomposing the KKT matrix into symmetric potential and skew-symmetric competitive components, we developed a Newton-type solver that uses the potential component as a preconditioner and accounts for competitive interactions through iterative refinement.
We established convergence guarantees, derived an a posteriori stopping criterion.
The symmetric structure further enables a reusable \ldlt{} factorization, reducing the cost of repeated linear solves.
Numerical benchmarks demonstrated improved runtime and scalability over existing methods, while experiments with eight miniature autonomous race cars validated the solver's real-time applicability.
%
% Future work will investigate less conservative certificates for potential dominance, adaptive refinement strategies, and extensions to larger-scale games with uncertainty and distributed information.

\bibliographystyle{ieeetran}
\bibliography{refs}
\balance

\appendix

\noindent
\textit{Proof of Theorem~\ref{thm:equivalent-lyapunov}}:
    For notational simplicity, we write $S$ and $A$ for $S(z)$ and $A(z)$.

    \noindent
    $\boxed{{(i)}\Leftrightarrow{(ii)}}$
    By definition, the game is $\gamma$-potential-dominant at $z$ if $\rho(S^{-1} A) < \gamma$.
    Equivalently, we have $\rho(\frac{1}\gamma S^{-1} A) < 1$, which holds, per the Lyapunov stability theory~\cite{lyapunov}, \emph{if and only if} there exists a positive definite matrix $\tilde P$ such that 
    \begin{equation}
        \label{eqn:lyapunov-inequality}
        \big(\tfrac{1}{\gamma} S^{-1} A \big)^\top \tilde P \big(\tfrac{1}{\gamma} S^{-1} A \big) - \tilde P \prec 0.
    \end{equation}
    Let $P = S^{-\top} \tilde P S^{-1}$.
    Since $S$ is invertible, $P$ is also positive definite, and~\eqref{eqn:lyapunov-inequality} can be re-written as 
    \[
        \tfrac{1}{\gamma^2} A^\top P A - S^{\top} \tilde P S \prec 0.
    \]
    Condition~$(ii)$ then follows immediately.
    \vspace{+5pt}

    \noindent
    $\boxed{{(ii)}\Leftrightarrow{(iii)}}$
    We now prove the equivalence between the linear matrix inequality and the induced $P$-norm bound.
    
    First, we prove the direction
    $(ii) ~\Rightarrow~ (iii)$: 
    Assume Condition~(ii) holds, i.e., there exists $P\succ 0$ such that $A^\top P A - \gamma ^2 S^\top P S \prec 0$.
    By negative definiteness,
    \begin{equation}
        \label{apdx-eqn:000}
        x^\top \Big( A^\top P A - \gamma^2 S^\top P S \Big) x < 0, \quad \forall x \neq 0.
    \end{equation}
    Let $Q = S^\top P S$. Since $P\succ 0$ and $S$ is invertible, we have $Q \succ 0$.
    Rearranging~\eqref{apdx-eqn:000} yields $(S^{-1}Ax)^\top Q (S^{-1}Ax) < \gamma^2 x^\top Q x$, which simplifies to $\norm{Q}{S^{-1}Ax}^2 < \gamma^2 \norm{Q}{x}^2$ for all $x \ne 0$. 
    As a result, $\norm{Q}{S^{-1}A} \!=\! \sup_{\Vert x\Vert_Q=1} \norm{Q}{S^{-1}Ax} < \gamma$.
    
    \vspace{+5pt}
    To prove $(iii) ~\Rightarrow~ (ii)$, 
    assume that condition~$(iii)$ holds, i.e., there exists $Q\succ 0$ such that $\norm{Q}{S^{-1}A} < \gamma$. 
    Then, for all $x \neq 0$, $\norm{Q}{S^{-1}Ax} < \gamma \norm{Q}{x}$. 
    Squaring both sides yields
    \[
        x^\top (S^{-1}A)^\top Q (S^{-1}A) x < \gamma^2 x^\top Q x, \quad \forall \, x \ne 0.
    \]
    Thus, $(S^{-1}A)^\top Q (S^{-1}A) \prec \gamma^2 Q$. 
    Let $P=S^{-\top}QS^{-1}$, so that $Q=S^\top P S$. 
    Then we have
    \[
    A^\top P A-\gamma^2 S^\top P S\prec0,
    \]
    which is condition~$(ii)$.
    \hfill$\square$

\end{document}